\documentclass[sigconf,screen]{acmart}

\renewcommand\footnotetextcopyrightpermission[1]{}
\acmDOI{}
\acmISBN{}
\acmConference[]{}{}{}
\acmYear{2026}
\acmPrice{}

\usepackage{listings}
\usepackage{tikz}
\usetikzlibrary{arrows.meta,positioning,decorations.markings}

\newcommand{\dotop}{\cdot}
\newcommand{\FRM}{\textsf{FRM}}
\newcommand{\ICP}{\textsf{ICP}}
\newcommand{\Fin}{\mathrm{Fin}}
\newcommand{\core}{\mathrm{core}}

\theoremstyle{remark}
\newtheorem{remark}{Remark}

\begin{document}

\title{Pairwise Independence of Representation, Classification, and Composition in Finite Extensional Magmas}
\thanks{Preprint version. The Lean proofs, SAT search scripts, and manuscript text were developed with assistance from Claude (Anthropic). All results have been independently verified: Lean proofs compile with \texttt{lake build} (zero \texttt{sorry}), and SAT counterexamples are frozen in the repository for reproducibility. The author takes full responsibility for all contents.}

\author{Stefano Palmieri}
\affiliation{
  \institution{Independent Researcher}
  \city{Austin}
  \state{Texas}
  \country{USA}
}


\begin{abstract}
Nontrivial combinatory algebras with S and K must be infinite. Associativity is incompatible with combining a classifier and a retraction pair in a finite extensional magma. These obstructions exclude several standard settings from the finite extensional framework studied here, most notably nontrivial finite S+K-style combinatory algebras and associative structures (semigroups, monoids, groups, rings) carrying both a classifier and a retraction pair. What algebraic structure exists in the remaining landscape: finite, non-associative, total?

We identify three properties of finite extensional 2-pointed magmas: self-representation~(R), the classifier dichotomy~(D), and the Internal Composition Property~(H). We prove they are pairwise independent. Lean-verified finite counterexamples at sizes $4$ through $10$ establish all six non-implications, four with provably tight bounds. The minimum coexistence witness has $N\!=\!5$, which is optimal: ICP requires 3 pairwise distinct core elements, so $N \geq 5$. The three-category decomposition induced by D is an isomorphism invariant, and the ICP is logically equivalent to the standard Compose+Inert axioms. All results are formalized in Lean~4 with zero \texttt{sorry}.
\end{abstract}

\maketitle

% ══════════════════════════════════════════════════════════════
\section{Introduction}
\label{sec:introduction}
% ══════════════════════════════════════════════════════════════

A finite magma is the Cayley table of a binary operation, stripped of syntax and types. Reflective systems combine encoding, classification, and composition~\cite{smith84,brown16}. Are these algebraically independent, or does one force another? We study this question in finite extensional 2-pointed magmas (a finite set with a binary operation, exactly two left-absorbers, and extensionality; the framework follows universal algebra~\cite{burris81}).

Two obstructions confirm this is the right setting. Nontrivial combinatory algebras with $\mathbf{s}$ and $\mathbf{k}$ are necessarily infinite (Theorem~\ref{thm:k-infinity}), ruling out the most direct finite analogues of PCA/TCA-style~\cite{vanoosten08,longley15,bethke88} and Turing-object-style~\cite{cockett08} universality. Associativity is incompatible with combining a classifier and a retraction pair (Theorem~\ref{thm:no-assoc}), excluding semigroups, monoids, groups, and rings. Finite extensional 2-pointed magmas are what remains. We identify three algebraic properties corresponding to these components: self-representation~(R), the existence of a section-retraction pair on the non-absorber elements (the \emph{core}); the classifier dichotomy~(D), a clean partition of core elements into those whose action on core lands entirely in the absorbers and those whose action avoids them entirely; and the Internal Composition Property~(H), the existence of a non-trivial factorization in the left-regular representation restricted to core.

Our main result is that R, D, and H are pairwise independent: no one implies any other. Lean-verified finite counterexamples at sizes $N\!=\!4, 5, 6, 8,$ and $10$ establish all six non-implications; coexistence witnesses at $N\!=\!5$ and $N\!=\!6$ show all three can hold simultaneously. The $N\!=\!5$ bound is optimal: ICP requires 3 pairwise distinct core elements, forcing $|S| \geq 5$. Each counterexample is a concrete Cayley table found by SAT search and independently verified in Lean~4~\cite{moura21}, zero \texttt{sorry}. Whether the independence extends to $\lambda$-calculus or typed settings is an open question (Section~\ref{sec:discussion}).

Beyond independence, we prove several additional structural results. The no-associativity theorem (Theorem~\ref{thm:no-assoc}) rules out the classical algebraic hierarchy: no extensional 2-pointed magma carrying both a classifier and a retraction pair has a right identity or is associative. The three-category decomposition induced by~D is an algebraic invariant (Theorem~\ref{thm:functoriality}), and the ICP is equivalent to the standard Compose+Inert axioms (Theorem~\ref{thm:icp-equiv}).

\paragraph{Contributions.}
\begin{enumerate}
  \item Universal algebraic obstructions that motivate the setting: K-infinity rules out finite combinatory algebras; no-associativity rules out the classical algebraic hierarchy (Section~\ref{sec:universal}).
  \item An independence theorem: R, D, and H are pairwise independent, with Lean-verified counterexamples in all six directions at $N\!=\!4$ through $N\!=\!10$ (Section~\ref{sec:independence}), and coexistence witnesses at $N\!=\!5$ and $N\!=\!6$ (Section~\ref{sec:bound}).
  \item An optimal bound: $N\!=\!5$ is the minimum carrier size for R+D+H coexistence. ICP is impossible at $N\!=\!4$ by pigeonhole (Theorem~\ref{thm:no-icp-4}).
  \item Decomposition invariance: the three-category partition is preserved by isomorphisms (Section~\ref{sec:invariance}).
  \item A characterization of H confirming the ICP is not ad hoc: it is logically equivalent to Compose+Inert under variable renaming, with necessity of the non-degeneracy conditions demonstrated by Lean-verified witnesses (Section~\ref{sec:icp}).
\end{enumerate}

% ══════════════════════════════════════════════════════════════
\section{Definitions}
\label{sec:definitions}
% ══════════════════════════════════════════════════════════════

Throughout, $(S, \dotop)$ is a finite magma on $\Fin(n) = \{0, 1, \ldots, n-1\}$ with binary operation $\dotop : \Fin(n) \times \Fin(n) \to \Fin(n)$.

\begin{definition}[Extensional 2-Pointed Magma]
\label{def:epm}
An \emph{extensional 2-pointed magma} is a tuple $(S, \dotop, z_1, z_2)$ where:
\begin{enumerate}
  \item $z_1, z_2 \in S$ are distinct \emph{left-absorbers}: $z_i \dotop x = z_i$ for all $x$.
  \item No other element is a left-absorber.
  \item \emph{Extensionality}: if $a \dotop x = b \dotop x$ for all $x$, then $a = b$.
\end{enumerate}
We write $\core(S) := S \setminus \{z_1, z_2\}$ for the non-absorber elements.
\end{definition}

\begin{definition}[Faithful Retract Magma]
\label{def:frm}
A \emph{Faithful Retract Magma} $(\FRM)$ extends an extensional 2-pointed magma with a \emph{retraction pair} $s, r \in \core(S)$ satisfying $r \dotop (s \dotop x) = x$ for all $x \in \core(S)$, with $r \dotop z_1 = z_1$ (anchoring). The remaining absorber outputs of $s$ and $r$ are unconstrained.
\end{definition}

\begin{remark}[Retraction pair convention]
\label{rem:mutual}
Throughout this paper, the retraction pair $(s, r)$ satisfies mutual inverse on core: $r \dotop (s \dotop x) = x$ and $s \dotop (r \dotop x) = x$ for all $x \in \core(S)$. This matches the Lean formalization. Most theorems (9 of 11) require only the one-sided condition $r \dotop (s \dotop x) = x$; two results, retraction pair placement (both $s, r \in N$; Theorem~\ref{thm:placement}) and the tight bound $|S| \geq 5$ when $s \neq r$ (Theorem~\ref{thm:card}), use the full mutual inverse. We adopt the stronger convention for uniformity and note the weaker hypotheses where relevant.
\end{remark}

The retraction pair provides a section/retraction structure on core: $s$ acts as a section and $r$ as a retraction. In finite witnesses this can be used to organize elements into coding schemes, but we do not assume a canonical global encoding.

\begin{definition}[Classifier Dichotomy]
\label{def:cd}
An extensional 2-pointed magma satisfies the \emph{classifier dichotomy} when:
\begin{enumerate}
  \item A \emph{classifier} $\tau \in \core(S)$ exists with $\tau \dotop x \in \{z_1, z_2\}$ for all $x$.
  \item For every $y \in \core(S)$, either
    \begin{itemize}
      \item $y \dotop x \in \{z_1, z_2\}$ for all $x \in \core(S)$ \quad ($y$ is a \emph{classifier}), or
      \item $y \dotop x \notin \{z_1, z_2\}$ for all $x \in \core(S)$ \quad ($y$ is a \emph{non-classifier}).
    \end{itemize}
  \item At least one non-classifier exists (excluding the degenerate case where the ``dichotomy'' has only one inhabited class).
\end{enumerate}
\end{definition}

The dichotomy creates a three-category decomposition $S = Z \sqcup C \sqcup N$ where $Z = \{z_1, z_2\}$ (absorbers), $C$ (classifiers), and $N$ (non-classifiers). No element can be partially classifying. We call outputs in $\{z_1, z_2\}$ \emph{absorber-valued}.

\begin{definition}[Internal Composition Property]
\label{def:icp}
A magma $(S, \dotop, z_1, z_2)$ satisfies the \emph{Internal Composition Property} $(\ICP)$ when:
\[
\exists\, a, b, c \in \core(S),\ \text{pairwise distinct, such that:}
\]
\begin{enumerate}
  \item \emph{Core-preservation}: $b \dotop x \notin \{z_1, z_2\}$ for all $x \in \core(S)$.
  \item \emph{Factorization}: $a \dotop x = c \dotop (b \dotop x)$ for all $x \in \core(S)$.
  \item \emph{Non-triviality}: $|\{a \dotop x : x \in \core(S)\}| \geq 2$.
\end{enumerate}
\end{definition}

ICP says: the left-regular representation admits a non-trivial factorization on core. One core action factors through two others, one of them core-preserving: a single fragment of internal composition. In Section~\ref{sec:icp} we prove that $\ICP$ is equivalent to the standard operational axioms Compose+Inert (Theorem~\ref{thm:icp-equiv}).

\begin{definition}[Capability R]
An extensional 2-pointed magma has capability $R$ (self-representation) if it admits a retraction pair.
\end{definition}

\begin{definition}[Capability D]
An extensional 2-pointed magma has capability $D$ (self-description) if it satisfies the classifier dichotomy (Definition~\ref{def:cd}: existence of a classifier, the global dichotomy, and non-degeneracy).
\end{definition}

\begin{definition}[Capability H]
An extensional 2-pointed magma has capability $H$ (self-execution) if it satisfies the $\ICP$.
\end{definition}

All three capabilities are properties of the same base structure (Definition~\ref{def:epm}). None is defined as an extension of any other, making the independence theorem (Section~\ref{sec:independence}) a statement about six non-implications, all proved by Lean-verified counterexamples. The labels ``self-representation,'' ``self-description,'' and ``self-execution'' are heuristic names for the formal properties R, D, and H; no claim is made that these structures by themselves implement a full internal evaluator.

% ══════════════════════════════════════════════════════════════
\section{Universal Theorems}
\label{sec:universal}
% ══════════════════════════════════════════════════════════════

The following theorems are proved by pure algebraic reasoning (no \texttt{decide}) and formalized in \texttt{CatKripkeWallMinimal.lean} and \texttt{NoCommutativity.lean}. We state the hypotheses explicitly: some require only D, while others use both R and D.

\subsection{Theorems from D alone}

\begin{theorem}[Three-Category Decomposition]
\label{thm:three-cat}
In any extensional 2-pointed magma satisfying D, every element is a zero, a classifier, or a non-classifier. The classes are pairwise disjoint: in particular, $C \cap N = \emptyset$ (no element is partially classifying).
\end{theorem}

\begin{theorem}[Asymmetry]
\label{thm:asymmetry}
No extensional magma with two distinct left-absorbers is commutative.
\end{theorem}

\subsection{Theorems from R+D}

\begin{theorem}[Retraction Pair Placement]
\label{thm:placement}
In any extensional 2-pointed magma satisfying both R and D, any witness $r$ to R is a non-classifier: $r \in N$. (Under the stronger mutual-inverse condition in the Lean formalization, $s \in N$ also holds.)
\end{theorem}

\begin{theorem}[Cardinality Bounds]
\label{thm:card}
$|S| \geq 4$ for any extensional 2-pointed magma satisfying both R and D. (The bound $|S| \geq 4$ is nearly immediate from the definitions, since the four roles (two absorbers, one classifier, one non-classifier) must be distinct; the Lean-verified proof confirms that no subtle interaction reduces it. The bound is tight: a 4-element witness with $s = r$ is exhibited in Appendix~\ref{app:tables}.) Under the stronger mutual-inverse condition with $s \neq r$, $|S| \geq 5$ (also tight); this is the nontrivial bound, as it shows the retraction pair consumes two distinct core slots.
\end{theorem}

\begin{theorem}[No Right Identity]
\label{thm:no-rid}
No extensional 2-pointed magma satisfying both R and D has a right identity element.
\end{theorem}

\begin{theorem}[No Associativity]
\label{thm:no-assoc}
No extensional 2-pointed magma carrying both a classifier and a retraction pair is associative. In particular, no extensional 2-pointed magma satisfying both R and D is a semigroup, monoid, group, or ring: the classifier dichotomy with a retraction pair is incompatible with the entire classical algebraic hierarchy.
\end{theorem}

\begin{proof}
Assume $(\tau \dotop a) \dotop b = \tau \dotop (a \dotop b)$ for all $a, b$. We derive a contradiction in four steps. The proof uses only extensionality, the existence of a classifier $\tau$, and a retraction pair $(s, r)$. Neither the global dichotomy nor the non-classifier existence axiom is needed. Lean file: \texttt{CatKripkeWallMinimal.lean}.

\emph{Step 1: $\tau$ absorbs on the right.} For any $a, b$:
\[
\tau \dotop (a \dotop b)
= (\tau \dotop a) \dotop b
= z_i \dotop b
= z_i
= \tau \dotop a
\]
where the second equality uses $\tau \dotop a \in \{z_1, z_2\}$ (classifier), and the third uses left-absorption. So $\tau \dotop (a \dotop b) = \tau \dotop a$ for all $a, b$.

\emph{Step 2: $\tau$ is constant on core.} For any core $x$: the retraction pair gives $x = r \dotop (s \dotop x)$. Applying Step~1 with $a = r$, $b = s \dotop x$:
\[
\tau \dotop x = \tau \dotop (r \dotop (s \dotop x)) = \tau \dotop r.
\]
Set $v = \tau \dotop r \in \{z_1, z_2\}$. Then $\tau \dotop x = v$ for all core $x$.

\emph{Step 3: $\tau$ is constant on absorbers.} Since $\tau$ is a classifier, it hits both $z_1$ and $z_2$; otherwise, if $\tau \dotop x = z_1$ for all $x$, then $\tau$ and $z_1$ have identical rows, so $\tau = z_1$ by extensionality, contradicting $\tau \neq z_1$ (and symmetrically for $z_2$). Let $b_1, b_2$ witness $\tau \dotop b_1 = z_1$ and $\tau \dotop b_2 = z_2$. Applying Step~1 with $a = \tau$:
\[
\tau \dotop z_1 = \tau \dotop (\tau \dotop b_1) = \tau \dotop \tau = v,
\qquad
\tau \dotop z_2 = \tau \dotop (\tau \dotop b_2) = \tau \dotop \tau = v.
\]
The middle equalities use Step~1; $\tau \dotop \tau = v$ because $\tau$ is core (Step~2).

\emph{Step 4: Contradiction.} Steps 2 and 3 give $\tau \dotop x = v$ for all $x \in S$. Since $v \in \{z_1, z_2\}$, the classifier $\tau$ has the same row as absorber $v$. By extensionality, $\tau = v \in \{z_1, z_2\}$, contradicting $\tau \neq z_1, z_2$.
\end{proof}

\begin{remark}
The retraction pair is essential in Step~2: the identity $\tau \dotop x = \tau \dotop (r \dotop (s \dotop x)) = \tau \dotop r$ uses $r \dotop (s \dotop x) = x$ to show $\tau$ is constant on core. Without a retraction pair, the proof breaks, and the theorem does not extend to ``no extensional 2-pointed magma with a classifier is associative.''
\end{remark}

\subsection{The K-infinity obstruction}

The following classical result explains why the PCA/TCA framework cannot study finite algebras, and why our setting requires a different combinator basis.

\begin{theorem}[K-infinity]
\label{thm:k-infinity}
If $(A, \dotop, \mathbf{k}, \mathbf{s})$ is a combinatory algebra (total application satisfying $\mathbf{k} \dotop a \dotop b = a$ and the $\mathbf{s}$ equations) with $|A| \geq 2$, then $A$ is infinite.
\end{theorem}

\begin{proof}
The map $\varphi(x) = \mathbf{k} \dotop x$ is injective: $\mathbf{k} \dotop a = \mathbf{k} \dotop b$ implies $\mathbf{k} \dotop a \dotop c = \mathbf{k} \dotop b \dotop c$ for all $c$, hence $a = b$. We claim $\mathbf{k} \notin \mathrm{Im}(\varphi)$. Suppose $\mathbf{k} \dotop x = \mathbf{k}$ for some $x$. Then for all $y$, $\mathbf{k} \dotop x \dotop y = \mathbf{k} \dotop y$, so $x = \mathbf{k} \dotop y$ for all $y$. By injectivity, all $y$ are equal, contradicting $|A| \geq 2$. Thus $\varphi : A \to A$ is an injection missing $\mathbf{k}$, which is impossible for finite~$A$.
\end{proof}

\begin{remark}
The proof uses only the $\mathbf{k}$ equation ($\mathbf{k} \dotop a \dotop b = a$) and totality; $\mathbf{s}$ is irrelevant. Our extensional 2-pointed magmas avoid this obstruction: instead of a projector $\mathbf{k}$ with $\mathbf{k} \dotop a \dotop b = a$, they have left-absorbers $z$ with $z \dotop x = z$ (the dual pattern, constant on columns rather than rows) and a retraction pair $(s, r)$ with $r \dotop (s \dotop x) = x$ on core. Neither structure provides a non-surjective injection on the full ground set.
\end{remark}

% ══════════════════════════════════════════════════════════════
\section{Decomposition Invariance}
\label{sec:invariance}
% ══════════════════════════════════════════════════════════════

\begin{definition}[Absorber-Preserving Isomorphism]
An \emph{absorber-preserving isomorphism} $\varphi : M_1 \to M_2$ between extensional 2-pointed magmas is a bijection $\varphi : \Fin(n) \to \Fin(n)$ satisfying:
\begin{itemize}
  \item $\varphi(a \dotop_1 b) = \varphi(a) \dotop_2 \varphi(b)$ for all $a, b$
  \item $\varphi(z_1^{(1)}) = z_1^{(2)}$ and $\varphi(z_2^{(1)}) = z_2^{(2)}$
\end{itemize}
\end{definition}

\begin{theorem}[Functoriality]
\label{thm:functoriality}
Any absorber-preserving isomorphism between extensional 2-pointed magmas satisfying D preserves the three-category decomposition: $\varphi$ maps zeros to zeros, classifiers to classifiers, and non-classifiers to non-classifiers.
\end{theorem}

\begin{proof}[Proof sketch]
\emph{Zeros}: if $a$ is a left-absorber, then for all~$x$ in the image of~$\varphi$,
$\varphi(a) \dotop_2 \varphi(x) = \varphi(a \dotop_1 x) = \varphi(a)$.
By surjectivity, $\varphi(a)$ is a left-absorber.

\emph{Classifiers}: if $y \dotop_1 x \in \{z_1, z_2\}$ for all core~$x$, then
$\varphi(y) \dotop_2 \varphi(x) = \varphi(y \dotop_1 x) \in \{z_1', z_2'\}$
for all core~$\varphi(x)$.
By surjectivity and injectivity (to exclude absorber preimages), $\varphi(y)$ is a classifier.

\emph{Non-classifiers}: symmetric argument using injectivity to transport~$\notin \{z_1, z_2\}$.

Algebraic proof (no \texttt{decide}); Lean file: \texttt{Functoriality.lean}.
\end{proof}

\begin{corollary}[Canonicity of the decomposition]
\label{cor:canonicity}
The $Z/C/N$ partition is the unique coarsest non-trivial isomorphism-invariant partition of the core. Any absorber-preserving isomorphism preserves each class individually (Theorem~\ref{thm:functoriality}), so no coarser partition of core into two non-empty parts can be invariant unless it coincides with $C/N$. The classification is determined by a single algebraic invariant: whether an element's row on core has image contained in~$\{z_1, z_2\}$.
\end{corollary}

\begin{theorem}[Capability Invariance]
\label{thm:cap-invariance}
All three capabilities R, D, and H are preserved by absorber-preserving isomorphisms of extensional 2-pointed magmas. If $\varphi : M_1 \to M_2$ is an absorber-preserving isomorphism and $M_1$ satisfies R (resp.\ D, H), then $M_2$ satisfies R (resp.\ D, H). The three capabilities are intrinsic algebraic properties, not presentation artifacts.
\end{theorem}

\begin{proof}
R: if $(s, r)$ is a retraction pair in $M_1$, then $(\varphi(s), \varphi(r))$ is a retraction pair in $M_2$. The mutual-inverse equations transport through the homomorphism, and anchoring ($r \dotop z_1 = z_1$) transfers by zero-preservation.

D: the classifier $\varphi(\tau)$ is absorber-valued in $M_2$ (by the homomorphism and zero-preservation). The global dichotomy transfers by surjectivity: every element in $M_2$ has a preimage whose dichotomy status transports. Non-degeneracy transfers by injectivity.

H: the ICP triple $(\varphi(a), \varphi(b), \varphi(c))$ satisfies the factorization equation in $M_2$. Injectivity preserves pairwise-distinctness and non-triviality; surjectivity preserves core-preservation.

Algebraic proofs (no \texttt{decide}); Lean file: \texttt{CapabilityInvariance.lean}.
\end{proof}

% ══════════════════════════════════════════════════════════════
\section{ICP Characterization}
\label{sec:icp}
% ══════════════════════════════════════════════════════════════

An algebra with R can encode, decode, and look up any entry in the Cayley table via the retraction pair $(s, r)$. But this requires an \emph{external} machine to supply control flow: dispatch, composition, recursion. H internalizes a non-trivial composition pattern that would otherwise be supplied externally: dedicated elements whose left-multiplication implements composition ($\eta = \rho \circ g$) and storage ($g$). The $\ICP$ captures exactly this: one core action factors through two others, meaning the algebra contains a fragment of its own control flow as an element.

The standard formulation uses two operational axioms:
\begin{itemize}
  \item \emph{Compose}: $\eta \dotop x = \rho \dotop (g \dotop x)$ for all $x \in \core(S)$
  \item \emph{Inert}: $g \dotop x \notin \{z_1, z_2\}$ for all $x \in \core(S)$
\end{itemize}
with $\eta, g, \rho$ pairwise distinct non-absorbers and $\eta$ non-trivial ($\geq 2$ distinct values on core).

\begin{theorem}[$\ICP$ Equivalence]
\label{thm:icp-equiv}
For any finite magma on a 2-pointed set $(S, \dotop, z_1, z_2)$ of any size $n$:
\[
\ICP \iff \text{Compose} + \text{Inert} \text{ (with non-trivial composite)}.
\]
\end{theorem}

\begin{proof}
The bijection between witness triples is $(\eta, g, \rho) = (a, b, c)$: the composed element $a$ is $\eta$, the core-preserving element $b$ is $g$, and the outer element $c$ is $\rho$. Under this identification the factorization $a \dotop x = c \dotop (b \dotop x)$ becomes $\eta \dotop x = \rho \dotop (g \dotop x)$ (Compose), core-preservation of $b$ becomes Inert, and the remaining conditions (pairwise distinct, non-absorber) transfer directly.

The pairwise-distinctness and non-triviality conditions are not cosmetic; they block two degeneracies, and both are proved necessary by Lean-verified witnesses (\texttt{ICP.lean}). Without pairwise-distinctness, the retraction pair $(s, r)$ with $s = r$ trivially satisfies factorization via $s \dotop x = s \dotop (s \dotop x)$, giving every $\FRM$ a spurious $\ICP$; the minimal R{+}D model at $N\!=\!4$ (Cayley table in Appendix~\ref{app:tables}; Lean name \texttt{kripke4}) separates the weakened definition from the full one (theorem \texttt{icp\_distinct\_necessary}). Without non-triviality, any constant-on-core element $a$ (with $a \dotop x = v$ for all core $x$) trivially factors through any core-preserving $b$ by taking $c$ with $c \dotop y = v$; a custom 6-element $\FRM$ separates the weakened definition from the full one (theorem \texttt{icp\_nontrivial\_necessary}, \texttt{ICP.lean}).

The proof is pure equational logic, with no \texttt{decide}, no case analysis on $n$. It holds for every finite magma on a 2-pointed set, of any size. Lean file: \texttt{ICP.lean}, theorem \texttt{icp\_iff\_composeInert}.
\end{proof}

This equivalence gives H a single-concept definition: the left-regular representation contains a non-trivially composed element. The operational axioms (Compose, Inert) and the $\ICP$ are the same property under variable renaming. Additionally, $\ICP$ is invariant under $\FRM$ isomorphisms (algebraic proof, no \texttt{decide}; \texttt{ICP.lean}, theorem \texttt{FRMIso.preserves\_icp}).

% ══════════════════════════════════════════════════════════════
\section{Independence}
\label{sec:independence}
% ══════════════════════════════════════════════════════════════

\begin{figure}[t]
\centering
\begin{tikzpicture}[
    capability/.style={draw, rectangle, rounded corners=2pt, minimum size=9mm, font=\large\bfseries},
    arr/.style={-{Stealth[length=5pt]}, semithick},
    cross/.style={fill=white, draw=red!70!black, thick, inner sep=1pt, font=\scriptsize\bfseries},
    lbl/.style={font=\scriptsize, fill=white, inner sep=1pt},
  ]
  % Equilateral triangle, R at bottom center
  \node[capability] (R) at (0, 0) {R};
  \node[capability] (D) at (-2.6, 3.6) {D};
  \node[capability] (H) at ( 2.6, 3.6) {H};

  % D <-> H (top edge): upper arc and lower arc
  \draw[arr] (D.20) to[bend left=18] coordinate[midway] (DH) (H.160);
  \node[cross] at (DH) {$\times$};
  \node[lbl, above=1pt of DH] {$N\!=\!4$};

  \draw[arr] (H.200) to[bend left=18] coordinate[midway] (HD) (D.-20);
  \node[cross] at (HD) {$\times$};
  \node[lbl, below=1pt of HD] {$N\!=\!5$};

  % D <-> R (left edge): both bend left — outer arc goes left, inner goes right
  \draw[arr] (R.140) to[bend left=18] coordinate[midway] (RD) (D.280);
  \node[cross] at (RD) {$\times$};
  \node[lbl, left=1pt of RD] {$N\!=\!8$};

  \draw[arr] (D.310) to[bend left=18] coordinate[midway] (DR) (R.110);
  \node[cross] at (DR) {$\times$};
  \node[lbl, right=1pt of DR] {$N\!=\!4$};

  % H <-> R (right edge): mirror of left — both bend right
  \draw[arr] (R.40) to[bend right=18] coordinate[midway] (RH) (H.260);
  \node[cross] at (RH) {$\times$};
  \node[lbl, right=1pt of RH] {$N\!=\!6$};

  \draw[arr] (H.230) to[bend right=18] coordinate[midway] (HR) (R.70);
  \node[cross] at (HR) {$\times$};
  \node[lbl, left=1pt of HR] {$N\!=\!5$};
\end{tikzpicture}
\caption{Full independence structure. All six pairwise non-implications are proved by Lean-verified finite counterexamples at the indicated sizes. No capability implies any other.}
\label{fig:independence}
\end{figure}
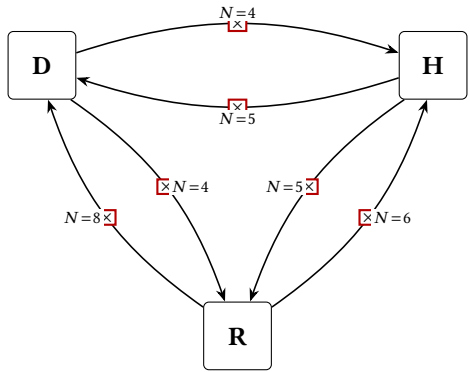

\begin{theorem}[Full Independence]
\label{thm:independence}
No capability implies any other. Specifically, all six pairwise non-implications hold:
\begin{enumerate}
  \item $R \not\Rightarrow D$: An extensional 2-pointed magma of size $8$ satisfying R contains a classifier element but fails the global dichotomy condition.
  \item $R \not\Rightarrow H$: An extensional 2-pointed magma of size $6$ satisfying R has 4 core elements (enough for $\ICP$ to be non-vacuous) but no triple satisfies the $\ICP$.
  \item $D \not\Rightarrow H$: An extensional 2-pointed magma of size $4$ satisfying D has only 2 core elements, so $\ICP$ is vacuously unsatisfiable (pigeonhole).
  \item $H \not\Rightarrow D$: An extensional 2-pointed magma of size $5$ satisfies $\ICP$ but violates the classifier dichotomy. This is tight: $\ICP$ requires $N \geq 5$.
  \item $D \not\Rightarrow R$: An extensional 2-pointed magma of size $4$ satisfies the classifier dichotomy but admits no retraction pair. This is tight: D requires $N \geq 4$.
  \item $H \not\Rightarrow R$: An extensional 2-pointed magma of size $5$ satisfies $\ICP$ but admits no retraction pair. This is tight: $\ICP$ requires $N \geq 5$.
\end{enumerate}
\end{theorem}

\begin{proof}
Each direction is proved by exhibiting a concrete Cayley table and verifying the claimed properties by \texttt{native\_decide} (or \texttt{decide} for $N \leq 8$). The tables are:

\paragraph{$R \not\Rightarrow D$ ($N\!=\!8$).}
An $8 \times 8$ Cayley table (Appendix~\ref{app:tables}) defining an $\FRM$ with $s = 2$, $r = 3$. Element~4 is a classifier, but element~5's row on core is $(6, 2, 1, 1, 1, 1)$, mixing absorber and non-absorber outputs, violating the dichotomy. SAT confirms the non-implication is achievable as early as $N\!=\!3$, but vacuously: with only 1 core element, the dichotomy cannot be formulated (it requires both a classifier and a non-classifier). The $N\!=\!8$ witness is the smallest with enough core structure for D to be non-vacuously falsifiable. Lean file: \texttt{Countermodel.lean}.

\paragraph{$R \not\Rightarrow H$ ($N\!=\!6$, structural).}
A $6 \times 6$ Cayley table (Appendix~\ref{app:tables}) defining an extensional 2-pointed magma with a retraction pair ($s = r = 2$, an involution on core) but no $\ICP$. The core has 4 elements, more than the 3 pairwise distinct witnesses $\ICP$ requires, yet no triple satisfies the $\ICP$. As with $R \not\Rightarrow D$, the non-implication is vacuously achievable at $N\!=\!3$ (only 1 core element), but the $N\!=\!6$ witness is the smallest where $\ICP$ failure is structural rather than dimensional. Lean file: \texttt{E2PM.lean}, theorem \texttt{s\_not\_\allowbreak implies\_\allowbreak icp\_\allowbreak structural}.

\paragraph{$D \not\Rightarrow H$ ($N\!=\!4$, tight).}
The minimal R+D model at $N\!=\!4$ (\texttt{kripke4} in Appendix~\ref{app:tables}) has only 2 core elements. $\ICP$ requires 3 pairwise distinct non-absorbers, so it is vacuously unsatisfiable. Lean file: \texttt{ICP.lean}, theorem \texttt{kripke4\_no\_icp}. The $N\!=\!10$ D+R model in \texttt{Countermodels10.lean} complements this: it has 8 core elements (more than enough for $\ICP$ to be non-vacuous), yet no triple satisfies $\ICP$. The formal non-implication holds at $N\!=\!4$; the $N\!=\!10$ witness shows it persists even when H has room to exist.

\paragraph{$H \not\Rightarrow D$ ($N\!=\!5$, tight).}
A $5 \times 5$ Cayley table satisfying $\ICP$ but with no classifier: all three core elements map core to core (no element produces only absorber-valued outputs on core). The bound is tight: $\ICP$ requires $N \geq 5$. Lean file: \texttt{E2PM.lean}, theorem \texttt{h\_not\_implies\_d\_tight}.

\paragraph{$D \not\Rightarrow R$ ($N\!=\!4$, tight).}
A $4 \times 4$ Cayley table with the classifier dichotomy (classifier at index~2, non-classifier at index~3) but no retraction pair among the 4 possible assignments from core. The bound is tight: D requires $N \geq 4$. Lean file: \texttt{E2PM.lean}, theorem \texttt{d\_not\_implies\_s\_tight}.

\paragraph{$H \not\Rightarrow R$ ($N\!=\!5$, tight).}
A $5 \times 5$ Cayley table satisfying $\ICP$ but admitting no retraction pair. The bound is tight: $\ICP$ requires $N \geq 5$. Lean file: \texttt{E2PM.lean}, theorem \texttt{h\_not\_implies\_s\_tight}.

All tables are generated by Z3 SAT solver and independently verified. The Cayley tables are frozen in the repository for reproducibility.
\end{proof}

\paragraph{Methodology and minimality.}
Each counterexample was found by encoding the relevant capability axioms (and the negation of the target capability) as a SAT problem over an $n \times n$ Cayley table, then solving with Z3, incrementing $n$ from 3 upward. Once found, each table was frozen and independently verified in Lean. Four of six directions have provably tight bounds: D$\not\Rightarrow$R and D$\not\Rightarrow$H are tight at $N\!=\!4$ (the minimum where D is non-vacuous, resp.\ where $\ICP$ is dimensionally impossible); H$\not\Rightarrow$R and H$\not\Rightarrow$D are tight at $N\!=\!5$ (the minimum where $\ICP$ can be stated). The remaining two (R$\not\Rightarrow$D at $N\!=\!8$, R$\not\Rightarrow$H at $N\!=\!6$) are the smallest non-vacuous witnesses found. The SAT search was exhaustive: all smaller sizes returned UNSAT for every role assignment tested. Whether different role assignments or overlapping encodings could yield tighter witnesses is open.

\paragraph{Scaling.}
The independence is not a small-size artifact. SAT search confirms that all six non-implications are satisfiable at every size from their respective minima through $N\!=\!15$ (the solver limit). No size was found where any direction becomes UNSAT. Whether the independence holds at all sizes (for all $n \geq 5$, there exist witnesses at size $n$ for each direction) is a natural conjecture supported by this data.

\paragraph{Interpretation.}
The independence result shows that R, D, and H are genuinely three things, not one. An algebra can encode and decode without classifying (R without D). An algebra can classify without encoding (D without R). An algebra can internally compose without classifying or encoding (H without R or D). No capability algebraically forces any other. In particular, the $H \not\Rightarrow D$ direction shows that an algebra can internalize a non-trivial composition pattern without satisfying the classifier dichotomy: internal composition does not force organizational separation.

% ══════════════════════════════════════════════════════════════
\section{Coexistence}
\label{sec:bound}
% ══════════════════════════════════════════════════════════════

\begin{theorem}[No ICP at $N\!=\!4$]
\label{thm:no-icp-4}
No extensional 2-pointed magma of size $4$ satisfies the $\ICP$. The core $\{2,3\}$ has only 2 elements, but $\ICP$ requires 3 pairwise distinct non-absorbers. Lean: \texttt{Witness5.lean}, \texttt{no\_icp\_at\_4}.
\end{theorem}

\begin{theorem}[Coexistence Witness]
\label{thm:bound}
The three capabilities R, D, and H are simultaneously satisfiable. A Lean-verified witness exists at $N\!=\!5$, which is optimal by Theorem~\ref{thm:no-icp-4}.
\end{theorem}

\begin{proof}
A concrete $5 \times 5$ Cayley table is exhibited with $s\!=\!r\!=\!2$ (identity on core), classifiers $\tau \in \{3,4\}$, and $\ICP$ witnessed by $a\!=\!3, b\!=\!2, c\!=\!4$. The retraction pair collapses to a single element ($s\!=\!r$) which doubles as the $\ICP$'s core-preserving map; since $b$ is the identity, the factorization $a \cdot x = c \cdot (b \cdot x)$ reduces to $a \cdot x = c \cdot x$ on core, yielding two classifiers that agree on core but distinguished by their absorber columns. The lower bound is pigeonhole: at $N\!=\!4$, the core $\{2,3\}$ has only 2 elements, but $\ICP$ requires 3 pairwise distinct non-absorbers (\texttt{no\_icp\_at\_4}). Lean file: \texttt{Witness5.lean}, theorem \texttt{sdh\_witness\_5}. A $6 \times 6$ witness with $s \neq r$ is also verified (\texttt{Witness6.lean}).
\end{proof}

\begin{proposition}[Separated-Role Witness]
\label{thm:separated}
There is a witness at $N\!=\!10$ in which the principal R+D+H roles (2 absorbers, 2 retraction pair, 1 classifier, 3 $\ICP$ witnesses) occupy 8 distinct elements.
\end{proposition}

\begin{proof}
A concrete $10 \times 10$ Cayley table (Appendix~\ref{app:tables}) satisfies all three capabilities with distinct role assignments. Lean file: \texttt{Witness10.lean}. A witness with $s \neq r$ (non-degenerate retraction) is exhibited at $N\!=\!6$ (Appendix~\ref{app:tables}).
\end{proof}

\paragraph{Minimum cardinality.}
$N\!=\!5$ is the optimal minimum for R+D+H coexistence (Lean-verified; \texttt{Witness5.lean}). The witness uses $s = r$ (identity on core), two classifiers, and one non-classifier: the retraction pair member doubles as the ICP's core-preserving element. At $N\!=\!4$ the core has only 2 elements but ICP requires 3 pairwise distinct non-absorbers, so $N\!=\!4$ is impossible (\texttt{no\_icp\_at\_4}).

The $N\!=\!5$ witness reveals how thin the three properties can become simultaneously. The retraction pair is the identity on core ($s = r$, no actual encoding work), and the ICP factorization $a \dotop x = c \dotop (b \dotop x)$ collapses to $a \dotop x = c \dotop x$ (row equality on core between two classifiers differing only on absorber columns). The three capabilities are formally present but barely interacting. The $N\!=\!6$ witness with $s \neq r$ (\texttt{Witness6.lean}) is the more informative example: the retraction pair does genuine encoding, the ICP factorization is non-degenerate, and the roles do not trivially overlap. We retain the general definitions (allowing $s = r$) since the algebraic results (independence, injectivity, no-associativity) hold regardless.

% ══════════════════════════════════════════════════════════════
\section{Related Work}
\label{sec:related}
% ══════════════════════════════════════════════════════════════

\paragraph{Partial and total combinatory algebras.}
PCAs~\cite{vanoosten08,longley15} and TCAs~\cite{bethke88} provide the standard algebraic framework for realizability and computability. PCA theory has a rich structural theory: morphisms between PCAs, the lattice of PCA quotients, and totalizability~\cite{bethke96} all involve detailed analysis of the application operation. PCA theory does study which elements realize which functions. However, this analysis concerns representability and inter-PCA relationships, not the \emph{partition} of elements into behavioral types by their row structure. This question requires a finite total operation table to formulate; it does not arise in the PCA/TCA setting, where application is partial or the carrier is infinite. Bethke and Klop~\cite{bethke96} proved that some PCAs cannot be extended to total ones, a structural result about the relationship between partiality and totality. The K-infinity theorem (our Theorem~\ref{thm:k-infinity}) shows that nontrivial S+K algebras must be infinite, explaining why finite algebras fall outside this framework and why our setting requires a different combinator basis.

\paragraph{Turing categories.}
Cockett and Hofstra~\cite{cockett08} define a Turing category as a cartesian restriction category with a Turing object~$A$ and a Turing morphism $\bullet : A \times A \to A$ that universally simulates all morphisms in the category. The Structure Theorem establishes that PCAs and Turing categories are two views of the same concept. Turing objects have cardinality ``one or infinite''~\cite{cockett08}, ruling out nontrivial finite Turing objects. The Turing morphism is described as ``the heart of the whole structure'' but its algebraic properties (row structure, element classification, associativity constraints) are not studied. Where Turing categories distribute computational roles across distinct objects (programs, data, evaluators), our finite algebras develop the same roles internally as emergent behavioral regions within a single carrier set: absorbers encode failure, classifiers provide absorber-valued judgment, and the retraction pair provides encoding.

\paragraph{Self-interpretation.}
Mogensen~\cite{mogensen92} constructs self-interpreters for the $\lambda$-calculus. Brown and Palsberg~\cite{brown16} show that self-interpretation is possible in strongly-normalizing typed languages, contradicting earlier impossibility results. These works study self-interpretation (our R) but do not address the independence of R from D or H.

\paragraph{Independence in finite algebra.}
Independence and non-implication results between equational properties have a long history in universal algebra, from Birkhoff's characterization of equational classes~\cite{burris81} to Mal'cev conditions characterizing implications between properties of varieties. Our results operate at a different level: we separate \emph{operational} properties (retraction pairs, classifier partitions, factorizations in the left-regular representation) rather than equational identities, and our base class, the finite extensional 2-pointed magmas, is not a variety in the Birkhoff sense (extensionality is a quasi-identity, not an equation). To our knowledge, the pairwise independence of retraction pairs, classifier dichotomies, and factorization properties in finite magmas has not been previously established. The SAT-find-then-Lean-verify methodology is related to the use of automated tools such as Mace4~\cite{mccune10} for counterexample search in algebra, though we use Z3 for constraint solving and Lean~4 for independent formal verification.

\paragraph{Finite model theory and machine-checked proofs.}
Independence results via finite counterexamples are standard in finite model theory~\cite{libkin04}. Each Cayley table is frozen in the Lean source and verified by the kernel, eliminating transcription errors. The sizes involved ($N\!=\!5$ through $N\!=\!10$) are small enough for \texttt{decide}/\texttt{native\_decide} but large enough that hand-verification of all conditions would be error-prone. The universal algebraic theorems (no-associativity, decomposition invariance, ICP equivalence) are proved by pure equational reasoning without \texttt{decide}, and hold for all finite sizes.

% ══════════════════════════════════════════════════════════════
\section{Discussion}
\label{sec:discussion}
% ══════════════════════════════════════════════════════════════

\paragraph{The classifier dichotomy is organizational.}
D organizes elements into coherent categories: classifiers that test and non-classifiers that transform. The $H \not\Rightarrow D$ result shows that an algebra can internalize composition without this organization. D contributes semantic structure rather than compositional behavior.

\paragraph{Why magmas?}
Magmas are not a limitation but a natural choice, for two independent reasons. First, the K-infinity theorem (Theorem~\ref{thm:k-infinity}) shows that nontrivial combinatory algebras with $\mathbf{s}$ and $\mathbf{k}$ must be infinite: the $\mathbf{k}$ combinator is a non-surjective injection on any finite set with $\geq 2$ elements. Finite self-description requires a different combinator basis, one without a K-like projector, and our extensional 2-pointed magmas provide exactly this. Second, Theorem~\ref{thm:no-assoc} shows that associativity forces any classifier to have a constant row, collapsing it to an absorber. Combined with the no-right-identity result (Theorem~\ref{thm:no-rid}), this independently rules out the entire classical algebraic hierarchy: no semigroup, monoid, group, or ring can support the classifier dichotomy alongside a retraction pair. Self-description lives outside both the PCA/TCA framework (by cardinality) and the classical algebraic hierarchy (by associativity). This connects to the broader observation that reflective systems are inherently non-associative: $\lambda$-calculus application is non-associative ($(f\,g)\,x \neq f\,(g\,x)$ in general). The no-associativity theorem shows this is algebraically forced, not a design accident.

\paragraph{The role of extensionality.}
Extensionality (distinct elements have distinct rows) is load-bearing throughout. Without it, the no-associativity proof fails: Step~4 derives $\tau = v$ from identical rows. Without it, any constant-row element trivially satisfies both classifier and non-classifier conditions, making the three-category decomposition degenerate. Extensionality is not merely a convenience but the condition that makes element-level classification meaningful. An element's identity \emph{is} its row in the Cayley table; extensionality formalizes this.

\paragraph{Why these three properties?}
Extensionality identifies each element with its row in the Cayley table. The elements of a finite extensional 2-pointed magma are therefore a family of typed maps $\{L_a^{\core} : \core \to \core \cup Z\}$, one per element, with the absorber maps $L_{z_1}, L_{z_2}$ being constant. R, D, and H correspond to the three fundamental structural questions about such a family:
\begin{itemize}
  \item \emph{Splitting} (R): does the identity on core split through the family? That is, do some $L_s, L_r$ satisfy $L_r \circ L_s = \mathrm{id}_{\core}$? This is a section-retraction pair~\cite{maclane71}.
  \item \emph{Codomain purity} (D): does each map hit $\core$ or $Z$, never both? This is the unique coarsest non-trivial isomorphism-invariant partition of core (Corollary~\ref{cor:canonicity}).
  \item \emph{Factorization} (H): does the family contain a non-trivial composition, $L_a = L_c \circ L_b$ with $L_b$ core-preserving and $L_a$ non-constant? This is the weakest factorization condition that is not vacuously satisfiable: weakening either the pairwise-distinctness or the non-triviality requirement makes every retraction-equipped magma satisfy it trivially (Lean-verified separations in \texttt{ICP.lean}).
\end{itemize}
Each property is minimal in its own sense. R is the standard notion of splitting. D is the coarsest invariant partition. H is the weakest non-degenerate factorization. The independence theorem confirms these are genuinely three questions, not one. The precise categorical status of D and H is discussed further in the categorical connections paragraph below; whether either corresponds to a known universal property remains open.

\paragraph{Categorical connections and their limits.}
Several constructions in this paper have natural categorical readings. The core $\core(S)$ is a retract of the carrier via $(s, r)$, in the standard sense of~\cite{maclane71}. The three-category decomposition defines a functor from the category $\mathbf{E2PM}$ (extensional 2-pointed magmas with absorber-preserving homomorphisms) to the category of three-block set partitions: Theorems~\ref{thm:functoriality} and~\ref{thm:cap-invariance} establish that both the decomposition and the three capabilities are preserved by morphisms. The classifier dichotomy is a codomain-purity condition on the family of left-multiplication maps $\{L_a^{\core} : \core \to \core \cup Z\}$, partitioning them by image type.

The categorical analogy breaks down at composition. The left-multiplication maps $L_a^{\core}$ are not closed under composition: $L_a \circ L_b$ need not equal $L_c$ for any element $c$. For instance, in the $N\!=\!5$ coexistence witness, $L_2 \circ L_3$ on core gives $(0, 2, 0)$, which is not the core row of any element. This means the row family does not form a transformation monoid, and standard categorical constructions (internal hom, enrichment, Turing morphism) cannot be stated internally. The no-associativity theorem (Theorem~\ref{thm:no-assoc}) predicts this failure: associativity of the magma operation would be a prerequisite for the left-multiplication maps to form a monoid, and the theorem shows associativity is incompatible with the R+D structure. The composition-closure failure is not a gap in the formalization but a consequence of the same obstruction that motivates the setting.

\paragraph{Why the independence holds.}
The independence of D and H reflects incompatible row constraints under finiteness. A classifier's row is fully committed to absorber values on core: every core input maps to $z_1$ or $z_2$. An ICP witness $a$ requires $a \dotop x = c \dotop (b \dotop x)$ with at least two distinct values on core, so its row must be non-trivially structured. These demands are incompatible for any single element, so different elements must carry these roles. The independence holds because finiteness limits the total supply of elements and constrains which row patterns can coexist under extensionality. The $H \not\Rightarrow D$ direction has the same character: an ICP triple $(a, b, c)$ constrains three rows to satisfy a factorization equation, but this imposes no absorber-output condition on any fourth element.

\paragraph{The infinite boundary.}
In infinite carriers where the set of left-multiplication maps is unconstrained, such as the full function space over a countable set, the row-scarcity constraint dissolves. The carrier is large enough to accommodate elements with arbitrary row patterns, and the tension between classifier rows and factorization rows need not arise. Whether the independence survives in specific infinite algebraic settings is open. This is an observation about resource limitations, not a theorem about Turing completeness: the precise conditions under which the independence transfers require defining the relevant algebraic properties for infinite carriers, which is beyond this paper's scope.

\paragraph{R as a separate design choice.}
The independence of R from D and H plausibly extends beyond finite magmas. A faithful internal encoding, a section-retraction pair on the carrier, is a separate design choice from classification or composition. Mogensen's self-interpreter for the $\lambda$-calculus~\cite{mogensen92} requires an explicit quoting convention; without it, the calculus is computationally universal but lacks internal self-representation. This suggests that the R independence is not an artifact of finiteness but reflects a genuine structural distinction between having computational capabilities and having a faithful encoding of one's own elements as data.

\paragraph{Scope and limitations.}
Whether the three-capability separation appears in other settings ($\lambda$-calculus, typed systems, infinite algebras) is the primary open question. Theorem~\ref{thm:no-assoc} constrains the target: any such translation must work without associativity or identity. Each capability admits multiple axiom forms; the independence result holds for all variants. The independence also implies that any system combining all three capabilities must introduce additional axioms beyond R, D, and H to connect them; the nature of these connecting axioms is a natural direction for further work.

% ══════════════════════════════════════════════════════════════
\section{Conclusion}
% ══════════════════════════════════════════════════════════════

Two obstructions, K-infinity and no-associativity, rule out both combinatory algebras and the classical algebraic hierarchy as settings for finite algebras combining encoding with classification. In the remaining landscape of finite extensional 2-pointed magmas, we have identified three natural algebraic properties, self-representation~(R), the classifier dichotomy~(D), and the Internal Composition Property~(H), and proved they are pairwise independent. The six non-implications are machine-checked in Lean~4. The optimal coexistence witness has $N\!=\!5$ (Theorem~\ref{thm:bound}). The three-category decomposition induced by D is a functorial invariant (Theorem~\ref{thm:functoriality}), and the ICP is logically equivalent to Compose+Inert (Theorem~\ref{thm:icp-equiv}). All results are formalized in Lean~4 (93 theorems across 12 files, zero \texttt{sorry}).

\paragraph{Open problems.}
\begin{enumerate}
  \item \emph{Cross-formalism universality}: Does the three-capability separation appear in the $\lambda$-calculus, typed settings, or infinite algebras? The no-associativity obstruction (Theorem~\ref{thm:no-assoc}) constrains the target setting.
  \item \emph{Categorical characterization of H}: The ICP is the weakest non-degenerate factorization property of the left-regular representation (weakening either pairwise-distinctness or non-triviality makes it vacuous). Is it equivalent to a standard universal property (e.g., existence of a partial internal hom)?
  \item \emph{Categorical characterization of D}: D is the unique coarsest non-trivial isomorphism-invariant partition of core (Corollary~\ref{cor:canonicity}), isolating the two-valued aspect of a subobject classifier without the universality aspect. The left-multiplication maps lack composition closure, precluding standard categorical formulations. Does D correspond to a known concept in non-associative categorical algebra, or does it constitute a new structural property?
\end{enumerate}

\paragraph{Artifact.}
All Lean files, counterexample tables, and SAT reproduction scripts are available at \url{https://github.com/stefanopalmieri/finite-magma-independence}.

% ══════════════════════════════════════════════════════════════
% References
% ══════════════════════════════════════════════════════════════

% ══════════════════════════════════════════════════════════════
\appendix
% ══════════════════════════════════════════════════════════════

\section{Cayley Tables}
\label{app:tables}

All tables are extracted from the Lean source files and verified by \texttt{lake build}. Element $i$'s row is row~$i$; entry $(i,j)$ is $i \dotop j$.

\paragraph{$N\!=\!4$: Minimal R+D witness ($s = r$).}
Roles: $0 = z_1$, $1 = z_2$, $2 = \tau$, $3 = s = r$. \quad Lean: \texttt{CatKripkeWallMinimal.lean}, \texttt{kripke4}.

\medskip
{\small
\begin{tabular}{c|cccc}
$\dotop$ & 0 & 1 & 2 & 3 \\\hline
0 & 0 & 0 & 0 & 0 \\
1 & 1 & 1 & 1 & 1 \\
2 & 0 & 1 & 0 & 1 \\
3 & 0 & 0 & 2 & 3 \\
\end{tabular}}

\paragraph{$N\!=\!5$: Minimal R+D witness ($s \neq r$).}
Roles: $0 = z_1$, $1 = z_2$, $2 = s$, $3 = r$, $4 = \tau$. \quad Lean: \texttt{CatKripkeWallMinimal.lean}, \texttt{kripke5}.

\medskip
{\small
\begin{tabular}{c|ccccc}
$\dotop$ & 0 & 1 & 2 & 3 & 4 \\\hline
0 & 0 & 0 & 0 & 0 & 0 \\
1 & 1 & 1 & 1 & 1 & 1 \\
2 & 1 & 0 & 3 & 4 & 2 \\
3 & 0 & 2 & 4 & 2 & 3 \\
4 & 0 & 1 & 1 & 0 & 0 \\
\end{tabular}}

\paragraph{$N\!=\!5$: R+D+H Optimal Coexistence ($s = r$, minimal coexistence witness).}
Roles: $0 = z_1$, $1 = z_2$, $2 = s = r$ (identity on core, ICP $b$), $3 = \tau$ (classifier, ICP $a$), $4 = \tau_2$ (classifier, ICP $c$). \quad Lean: \texttt{Witness5.lean}.

\medskip
{\small
\begin{tabular}{c|ccccc}
$\dotop$ & 0 & 1 & 2 & 3 & 4 \\\hline
0 & 0 & 0 & 0 & 0 & 0 \\
1 & 1 & 1 & 1 & 1 & 1 \\
2 & 0 & 2 & 2 & 3 & 4 \\
3 & 0 & 0 & 0 & 1 & 0 \\
4 & 0 & 1 & 0 & 1 & 0 \\
\end{tabular}}

\paragraph{$N\!=\!6$: R+D+H Coexistence (role overlap).}
Roles: $0 = z_1$, $1 = z_2$, $2 = s$ (also ICP $a$), $3 = r$ (also ICP $b$), $4 = \tau$, $5 =$ ICP $c$. \quad Lean: \texttt{Witness6.lean}.

\medskip
{\small
\begin{tabular}{c|cccccc}
$\dotop$ & 0 & 1 & 2 & 3 & 4 & 5 \\\hline
0 & 0 & 0 & 0 & 0 & 0 & 0 \\
1 & 1 & 1 & 1 & 1 & 1 & 1 \\
2 & 3 & 3 & 4 & 2 & 5 & 3 \\
3 & 0 & 1 & 3 & 5 & 2 & 4 \\
4 & 0 & 0 & 1 & 0 & 1 & 1 \\
5 & 2 & 2 & 5 & 4 & 3 & 2 \\
\end{tabular}}

\paragraph{$N\!=\!8$: $R \not\Rightarrow D$ Counterexample.}
Roles: $0 = z_1$, $1 = z_2$, $2 = s$, $3 = r$, $4 =$ classifier, $5, 7 =$ mixed (violate dichotomy). \quad Lean: \texttt{Countermodel.lean}.

\medskip
{\small
\begin{tabular}{c|cccccccc}
$\dotop$ & 0 & 1 & 2 & 3 & 4 & 5 & 6 & 7 \\\hline
0 & 0 & 0 & 0 & 0 & 0 & 0 & 0 & 0 \\
1 & 1 & 1 & 1 & 1 & 1 & 1 & 1 & 1 \\
2 & 3 & 3 & 7 & 3 & 4 & 6 & 5 & 2 \\
3 & 0 & 1 & 7 & 3 & 4 & 6 & 5 & 2 \\
4 & 0 & 0 & 0 & 0 & 0 & 0 & 1 & 0 \\
5 & 6 & 2 & 6 & 2 & 1 & 1 & 1 & 1 \\
6 & 0 & 0 & 5 & 2 & 2 & 2 & 2 & 2 \\
7 & 2 & 2 & 2 & 1 & 2 & 2 & 6 & 3 \\
\end{tabular}}

\paragraph{$N\!=\!6$: $R \not\Rightarrow H$ Counterexample (structural).}
Roles: $0 = z_1$, $1 = z_2$, $2 = s = r$ (involution on core). Retraction pair holds; no triple satisfies $\ICP$ despite 4 core elements. \quad Lean: \texttt{E2PM.lean} (\texttt{sNoH6\_e2pm}).

\medskip
{\small
\begin{tabular}{c|cccccc}
$\dotop$ & 0 & 1 & 2 & 3 & 4 & 5 \\\hline
0 & 0 & 0 & 0 & 0 & 0 & 0 \\
1 & 1 & 1 & 1 & 1 & 1 & 1 \\
2 & 0 & 3 & 3 & 2 & 5 & 4 \\
3 & 2 & 4 & 5 & 5 & 1 & 4 \\
4 & 5 & 3 & 0 & 0 & 3 & 2 \\
5 & 4 & 2 & 2 & 2 & 2 & 2 \\
\end{tabular}}

\paragraph{$N\!=\!10$: $D \not\Rightarrow H$ Counterexample.}
Roles: $0 = z_1$, $1 = z_2$, $2 = s$, $3 = r$, $4 = \tau$. No triple satisfies $\ICP$. \quad Lean: \texttt{Countermodels10.lean} (\texttt{dNotH}), \texttt{ICP.lean} (\texttt{dNotH\_no\_icp}).

\medskip
{\footnotesize
\begin{tabular}{c|cccccccccc}
$\dotop$ & 0 & 1 & 2 & 3 & 4 & 5 & 6 & 7 & 8 & 9 \\\hline
0 & 0 & 0 & 0 & 0 & 0 & 0 & 0 & 0 & 0 & 0 \\
1 & 1 & 1 & 1 & 1 & 1 & 1 & 1 & 1 & 1 & 1 \\
2 & 3 & 3 & 2 & 3 & 4 & 5 & 6 & 7 & 9 & 8 \\
3 & 0 & 1 & 2 & 3 & 4 & 5 & 6 & 7 & 9 & 8 \\
4 & 0 & 0 & 1 & 1 & 1 & 1 & 1 & 1 & 1 & 1 \\
5 & 1 & 0 & 0 & 0 & 0 & 0 & 0 & 0 & 0 & 0 \\
6 & 2 & 3 & 9 & 9 & 9 & 9 & 9 & 9 & 9 & 8 \\
7 & 3 & 2 & 9 & 9 & 9 & 9 & 9 & 9 & 9 & 8 \\
8 & 1 & 0 & 1 & 0 & 1 & 1 & 1 & 1 & 0 & 0 \\
9 & 0 & 1 & 0 & 1 & 0 & 1 & 1 & 0 & 1 & 1 \\
\end{tabular}}

\paragraph{$N\!=\!10$: $H \not\Rightarrow D$ Counterexample.}
Roles: $0 = z_1$, $1 = z_2$, $2 = s$, $3 = r$, $4 = \tau$. $\ICP$ holds ($a\!=\!8, b\!=\!6, c\!=\!7$). Element~5 violates dichotomy. \quad Lean: \texttt{Countermodels10.lean} (\texttt{hNotD}).

\medskip
{\footnotesize
\begin{tabular}{c|cccccccccc}
$\dotop$ & 0 & 1 & 2 & 3 & 4 & 5 & 6 & 7 & 8 & 9 \\\hline
0 & 0 & 0 & 0 & 0 & 0 & 0 & 0 & 0 & 0 & 0 \\
1 & 1 & 1 & 1 & 1 & 1 & 1 & 1 & 1 & 1 & 1 \\
2 & 3 & 1 & 3 & 4 & 9 & 6 & 8 & 5 & 7 & 2 \\
3 & 0 & 1 & 9 & 2 & 3 & 7 & 5 & 8 & 6 & 4 \\
4 & 0 & 0 & 1 & 1 & 1 & 1 & 1 & 1 & 0 & 0 \\
5 & 0 & 0 & 2 & 0 & 0 & 0 & 0 & 0 & 3 & 1 \\
6 & 2 & 2 & 2 & 8 & 3 & 9 & 4 & 7 & 9 & 7 \\
7 & 8 & 3 & 2 & 8 & 3 & 9 & 4 & 7 & 3 & 1 \\
8 & 9 & 2 & 2 & 3 & 8 & 1 & 3 & 7 & 1 & 7 \\
9 & 2 & 2 & 2 & 2 & 4 & 7 & 6 & 7 & 2 & 0 \\
\end{tabular}}

\paragraph{$N\!=\!4$: $D \not\Rightarrow R$ Counterexample (tight).}
Element~2 is a classifier: row on core is $(1,1) \subseteq \{z_1, z_2\}$. Element~3 is a non-classifier: row on core is $(2,2) \subseteq \core$. No retraction pair exists among the 4 candidate $(s,r)$ assignments. Tight: D requires $N \geq 4$. \quad Lean: \texttt{E2PM.lean} (\texttt{d\_not\_implies\_s\_tight}).

\medskip
{\small
\begin{tabular}{c|cccc}
$\dotop$ & 0 & 1 & 2 & 3 \\\hline
0 & 0 & 0 & 0 & 0 \\
1 & 1 & 1 & 1 & 1 \\
2 & 0 & 1 & 1 & 1 \\
3 & 2 & 3 & 2 & 2 \\
\end{tabular}}

\paragraph{$N\!=\!5$: $H \not\Rightarrow R$ Counterexample (tight).}
$\ICP$ holds but no retraction pair exists. Tight: $\ICP$ requires $N \geq 5$. \quad Lean: \texttt{E2PM.lean} (\texttt{h\_not\_implies\_s\_tight}).

\medskip
{\small
\begin{tabular}{c|ccccc}
$\dotop$ & 0 & 1 & 2 & 3 & 4 \\\hline
0 & 0 & 0 & 0 & 0 & 0 \\
1 & 1 & 1 & 1 & 1 & 1 \\
2 & 3 & 1 & 0 & 3 & 1 \\
3 & 2 & 4 & 3 & 4 & 2 \\
4 & 2 & 2 & 1 & 0 & 3 \\
\end{tabular}}

\paragraph{$N\!=\!5$: $H \not\Rightarrow D$ Counterexample (tight).}
$\ICP$ holds but the dichotomy fails: no element maps core entirely to $\{z_1, z_2\}$ (no classifier exists). Tight: $\ICP$ requires $N \geq 5$. \quad Lean: \texttt{E2PM.lean} (\texttt{h\_not\_implies\_d\_tight}).

\medskip
{\small
\begin{tabular}{c|ccccc}
$\dotop$ & 0 & 1 & 2 & 3 & 4 \\\hline
0 & 0 & 0 & 0 & 0 & 0 \\
1 & 1 & 1 & 1 & 1 & 1 \\
2 & 3 & 3 & 4 & 3 & 3 \\
3 & 2 & 4 & 4 & 4 & 3 \\
4 & 2 & 2 & 2 & 4 & 4 \\
\end{tabular}}

\paragraph{$N\!=\!10$: R+D+H Coexistence Witness (all roles distinct).}
R+D+H roles: $0 = z_1$, $1 = z_2$, $2 = s$, $3 = r$, $4 = \tau$, $8 = \eta$ (ICP $a$), $6 = g$ (ICP $b$, core-preserving), $7 = \rho$ (ICP $c$). \quad Lean: \texttt{Witness10.lean}.

\medskip
{\footnotesize
\begin{tabular}{c|cccccccccc}
$\dotop$ & 0 & 1 & 2 & 3 & 4 & 5 & 6 & 7 & 8 & 9 \\\hline
0 & 0 & 0 & 0 & 0 & 0 & 0 & 0 & 0 & 0 & 0 \\
1 & 1 & 1 & 1 & 1 & 1 & 1 & 1 & 1 & 1 & 1 \\
2 & 3 & 3 & 4 & 3 & 7 & 5 & 9 & 6 & 8 & 2 \\
3 & 0 & 1 & 9 & 3 & 2 & 5 & 7 & 4 & 8 & 6 \\
4 & 0 & 0 & 1 & 1 & 1 & 0 & 0 & 0 & 1 & 1 \\
5 & 2 & 2 & 7 & 2 & 8 & 9 & 4 & 3 & 4 & 2 \\
6 & 0 & 0 & 6 & 4 & 8 & 7 & 3 & 3 & 4 & 9 \\
7 & 2 & 2 & 6 & 4 & 8 & 9 & 4 & 3 & 4 & 9 \\
8 & 2 & 2 & 4 & 8 & 4 & 3 & 4 & 4 & 8 & 9 \\
9 & 3 & 4 & 7 & 3 & 9 & 2 & 2 & 9 & 2 & 3 \\
\end{tabular}}

\section{Proof Inventory}
\label{app:inventory}

All Lean files compile with \texttt{lake build}, zero \texttt{sorry}. Proof styles: \emph{Algebraic} = pure equational reasoning (no \texttt{decide}); \emph{decide} = verified by kernel computation ($N \leq 8$); \emph{native\_decide} = verified by compiled native code ($N = 10$).

\medskip
{\small
\begin{tabular}{lrcl}
\textbf{File} & \textbf{Thms} & \textbf{Style} & \textbf{Content} \\
\hline
CatKripkeWallMinimal & 17 & Algebraic & Decomposition, bounds, wall \\
NoCommutativity      &  3 & Algebraic & Asymmetry \\
Functoriality        &  4 & Algebraic & Decomposition invariance \\
CapabilityInvariance &  4 & Algebraic & R, D, H each invariant \\
ICP                  & 20 & Alg.+decide & ICP $\Leftrightarrow$ Compose+Inert \\
Countermodel         &  5 & decide & $R \not\Rightarrow D$ ($N\!=\!8$) \\
Countermodels10      &  9 & native\_dec. & $D \not\Rightarrow H$, $H \not\Rightarrow D$ \\
E2PM                 & 19 & decide & $D \not\Rightarrow R$, $H \not\Rightarrow R$, $H \not\Rightarrow D$, $R \not\Rightarrow H$ (tight) \\
Witness10            &  6 & native\_dec. & R+D+H at $N\!=\!10$ \\
Witness5             &  3 & decide & R+D+H at $N\!=\!5$, no ICP at $N\!=\!4$ \\
Witness6             &  3 & decide & R+D+H at $N\!=\!6$ \\
\hline
\textbf{Total}       & \textbf{93} & & \\
\end{tabular}}


\begin{thebibliography}{20}

\bibitem{smith84}
Brian Cantwell Smith.
\newblock Reflection and semantics in {Lisp}.
\newblock In \emph{POPL}, pages 23--35, 1984.

\bibitem{mogensen92}
Torben Mogensen.
\newblock Efficient self-interpretation in lambda calculus.
\newblock \emph{Journal of Functional Programming}, 2(3):279--291, 1992.

\bibitem{brown16}
Matt Brown and Jens Palsberg.
\newblock Breaking through the normalization barrier: A self-interpreter for {F-omega}.
\newblock In \emph{POPL}, pages 5--17, 2016.

\bibitem{maclane71}
Saunders Mac~Lane.
\newblock \emph{Categories for the Working Mathematician}.
\newblock Springer, 1971.

\bibitem{johnstone02}
Peter~T. Johnstone.
\newblock \emph{Sketches of an Elephant: A Topos Theory Compendium}, vol.~1.
\newblock Oxford University Press, 2002.

\bibitem{libkin04}
Leonid Libkin.
\newblock \emph{Elements of Finite Model Theory}.
\newblock Springer, 2004.

\bibitem{burris81}
Stanley Burris and H.~P. Sankappanavar.
\newblock \emph{A Course in Universal Algebra}.
\newblock Springer, 1981.

\bibitem{mccune10}
William McCune.
\newblock \emph{Prover9 and Mace4}.
\newblock \url{https://www.cs.unm.edu/~mccune/prover9/}, 2005--2010.

\bibitem{moura21}
Leonardo de~Moura and Sebastian Ullrich.
\newblock The {Lean}~4 theorem prover and programming language.
\newblock In \emph{CADE}, pages 625--635, 2021.

\bibitem{vanoosten08}
Jaap van~Oosten.
\newblock \emph{Realizability: An Introduction to its Categorical Side}.
\newblock Studies in Logic and the Foundations of Mathematics~152, Elsevier, 2008.

\bibitem{longley15}
John Longley and Dag Normann.
\newblock \emph{Higher-Order Computability}.
\newblock Theory and Applications of Computability, Springer, 2015.

\bibitem{bethke88}
Ingemarie Bethke.
\newblock \emph{Notes on Partial Combinatory Algebras}.
\newblock PhD thesis, Universiteit van Amsterdam, 1988.

\bibitem{bethke96}
Ingemarie Bethke and Jan~Willem Klop.
\newblock Collapsing partial combinatory algebras.
\newblock In \emph{Higher-Order Algebra, Logic, and Term Rewriting (HOA 1995)}, LNCS~1074, pages 16--48, Springer, 1996. \textsc{doi}: 10.1007/3-540-61254-8\_18.

\bibitem{cockett08}
J.~Robin~B. Cockett and Pieter Hofstra.
\newblock Introduction to {T}uring categories.
\newblock \emph{Annals of Pure and Applied Logic}, 156(2--3):183--209, 2008. \textsc{doi}: 10.1016/j.apal.2008.04.005.

\end{thebibliography}
\end{document}